\documentclass[pmlr]{jmlr}

\usepackage{cmap} 
\usepackage[T1]{fontenc}
\usepackage{microtype} 

\usepackage{booktabs}
\usepackage{multirow}
\usepackage{colortbl}
\usepackage{tikz}
\usepackage{algorithm}
\usepackage{algorithmic}

\usetikzlibrary{shapes, shapes.geometric, arrows.meta, positioning, fit, patterns, shadows, backgrounds, calc, decorations.pathreplacing, decorations.pathmorphing}

\newcommand{\todo}[1]{}
\newcommand{\teach}[1]{}

\jmlrproceedings{PMLR}{Proceedings of Machine Learning Research}

\title[Continuous Modal Logical Neural Networks]{Continuous Modal Logical Neural Networks: Modal Reasoning via Stochastic Accessibility}

\author{%
  \Name{Antonin Sulc}\\
  \addr Lawrence Berkeley National Laboratory, Berkeley, CA, USA,
  \Email{asulc@lbl.gov}
}

\begin{document}

\maketitle

\begin{abstract}
We propose Fluid Logic, a paradigm in which modal logical reasoning, temporal, epistemic, doxastic, deontic, is lifted from discrete Kripke structures to continuous manifolds via Neural Stochastic Differential Equations (Neural SDEs).  Each type of modal operator is backed by a dedicated Neural SDE, and nested formulas compose these SDEs in a single differentiable graph.  A key instantiation is Logic-Informed Neural Networks (LINNs): analogous to Physics-Informed Neural Networks (PINNs), LINNs embed modal logical formulas such as ($\Box$ bounded) and ($\Diamond$ visits\_lobe) directly into the training loss, guiding neural networks to produce solutions that are structurally consistent with prescribed logical properties, without requiring knowledge of the governing equations.

The resulting framework, Continuous Modal Logical Neural Networks (CMLNNs), yields several key properties: (i) stochastic diffusion prevents quantifier collapse ($\Box$ and $\Diamond$ differ), unlike deterministic ODEs; (ii) modal operators are entropic risk measures, sound with respect to risk-based semantics with explicit Monte Carlo concentration guarantees; (iii)SDE-induced accessibility provides structural correspondence with classical modal axioms; (iv) parameterizing accessibility through dynamics reduces memory from quadratic in world count to linear in parameters.

Three case studies demonstrate that Fluid Logic and LINNs can guide neural networks to produce consistent solutions across diverse domains: epistemic/doxastic logic (multi-robot hallucination detection), temporal logic (recovering the Lorenz attractor geometry from logical constraints alone), and deontic logic (learning safe confinement dynamics from a logical specification).
\end{abstract}


\vspace{-0.85em}\section{Introduction}
\label{sec:intro}

Modal Logical Neural Networks (MLNNs)~\citep{sulc2025mlnn} established differentiable $\Box$ and $\Diamond$ operators over discrete worlds. However, finite world sets fundamentally limit reasoning over continuous physics, shifting beliefs, and temporal trajectories.

This paper introduces Continuous Modal Logical Neural Networks (CMLNNs), lifting modal reasoning from discrete Kripke structures to continuous manifolds via Neural SDEs~\citep{li2020scalable}. The central insight is that each modal logic type---temporal, epistemic, doxastic, deontic---corresponds to a different stochastic process over the same state space. Nested formulas compose these processes: $B_a(\Box_{[0,T]}\text{safe})$ runs believed dynamics, then evaluates temporal necessity along each path.

We call this paradigm Fluid Logic: logical truth flows stochastically through a manifold according to learned dynamics.  Accessibility becomes stochastic reachability, $w'$ is accessible from $w$ to the degree that $w'$ lies within the distribution of SDE sample paths from $w$.

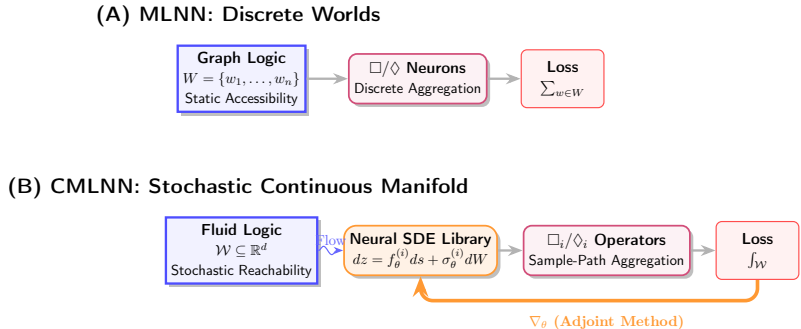
\begin{figure}[t!]
  \vspace{-1em}
  \centering
  \resizebox{0.7\linewidth}{!}{
  \begin{tikzpicture}[
      font=\sffamily,
      >=Stealth,
      node distance=0.6cm and 0.8cm,
      state_tensor/.style={
          rectangle, 
          draw=blue!70, 
          ultra thick, 
          fill=blue!5, 
          minimum width=2.4cm, 
          minimum height=1.6cm, 
          align=center, 
          drop shadow
      },
      sde_block/.style={
          rectangle, rounded corners=3mm, draw=orange!80, very thick, 
          fill=orange!10, minimum width=2.8cm, minimum height=1.2cm,
          align=center, drop shadow
      },
      modal_block/.style={
          rectangle, rounded corners=2mm, draw=purple!70, very thick, 
          fill=purple!5, minimum width=2.8cm, minimum height=1.2cm,
          align=center, drop shadow
      },
      loss_box/.style={
          rectangle, rounded corners, draw=red!70, thick, fill=red!5, 
          minimum width=2.0cm, minimum height=1.4cm, align=center
      },
      flow/.style={->, ultra thick, gray!60},
      gradient/.style={->, line width=2.5pt, orange!80, rounded corners=3mm},
      trajectory/.style={->, thick, blue!60, decorate, decoration={snake, amplitude=1mm, segment length=4mm}}
  ]
  
      
      \node[state_tensor] (DiscreteW) {
          \textbf{Graph Logic}\\
          {\small $W = \{w_1, \ldots, w_n\}$}\\
          {\small Static Accessibility}
      };
  
      \node[modal_block, right=1.0cm of DiscreteW] (DiscreteModal) {
          \textbf{$\Box/\Diamond$ Neurons}\\
          {\small Discrete Aggregation}
      };
      \draw[flow] (DiscreteW.east) -- (DiscreteModal.west);
  
      \node[loss_box, right=0.8cm of DiscreteModal] (LossA) {
          \textbf{Loss}\\
          {\small $\sum_{w \in W}$}
      };
      \draw[flow] (DiscreteModal.east) -- (LossA.west);
  
      \node[above=0.3cm of DiscreteW, align=center] (HeaderA) {
          \Large \textbf{(A) MLNN: Discrete Worlds}
      };
  
      
      \node[state_tensor, below=2.5cm of DiscreteW] (ContW) {
          \textbf{Fluid Logic}\\
          {\small $\mathcal{W} \subseteq \mathbb{R}^d$}\\
          {\small Stochastic Reachability}
      };
  
      \node[sde_block, right=0.6cm of ContW] (SDEBlock) {
          \textbf{Neural SDE Library}\\
          {\small $dz = f^{(i)}_\theta ds + \sigma^{(i)}_\theta dW$}
      };
      \draw[trajectory] (ContW.east) -- (SDEBlock.west) node[midway, above, font=\small] {Flow};
  
      \node[modal_block, right=0.6cm of SDEBlock] (ContModal) {
          \textbf{$\Box_i/\Diamond_i$ Operators}\\
          {\small Sample-Path Aggregation}
      };
      \draw[flow] (SDEBlock.east) -- (ContModal.west);
  
      \node[loss_box, right=0.6cm of ContModal] (LossB) {
          \textbf{Loss}\\
          {\small $\int_{\mathcal{W}}$}
      };
      \draw[flow] (ContModal.east) -- (LossB.west);
  
      \draw[gradient] (LossB.south) -- ++(0, -0.5) -| (SDEBlock.south);
      \node[text=orange!80, font=\bfseries\small, below=0.8cm of ContModal] {$\nabla_\theta$ (Adjoint Method)};
  
      \node[above=0.3cm of ContW, align=center] (HeaderB) {
          \Large \textbf{(B) CMLNN: Stochastic Continuous Manifold}
      };
  
  \end{tikzpicture}}
  \caption{The transition from Graph Logic to Fluid Logic. (A)~MLNNs operate over finite world sets with static, enumerable accessibility relations. (B)~CMLNNs embed worlds in a continuous manifold $\mathcal{W}$ where each modal operator $\Box_i, \Diamond_i$ is defined by its own Neural SDE.  The stochastic diffusion coefficient $\sigma^{(i)}_\theta$ generates genuine branching, giving $\Box$ (all sample paths) and $\Diamond$ (exists a sample path) semantically distinct meanings.  Gradients flow through the SDE solver via the stochastic adjoint method.}
\label{fig:graphical_abstract}
\vspace{-1.5em}
\end{figure}

The key advantages of Fluid Logic are: (1)~each operator type is backed by its own SDE, enabling multi-modal formulas impossible in STL; (2)~SDEs generate distributions over futures, giving $\Box$ and $\Diamond$ distinct semantics and preventing quantifier collapse (Theorem~\ref{thm:quantifier_noncollapse}); (3)~modal axioms emerge from SDE properties (initialization, reversibility, ergodicity) rather than explicit regularization; (4)~operators are entropic risk measures, sound with respect to risk-based modal semantics with MC concentration guarantees (Theorem~\ref{thm:soundness}).

\textbf{Contributions.}
(i)~We propose Fluid Logic, a paradigm where each modal operator type is backed by a dedicated Neural SDE, with nested formulas composing SDEs.
(ii)~We introduce Logic-Informed Neural Networks (LINNs), which embed modal logical specifications directly into the training loss, guiding neural networks to produce solutions consistent with prescribed logical properties.
(iii)~We prove that stochastic diffusion prevents quantifier collapse ($\Box \neq \Diamond$) and that the operators are sound w.r.t.\ entropic risk-based semantics.
(iv)~We demonstrate CMLNNs on three case studies spanning epistemic, doxastic, temporal, and deontic modal logics.

\vspace{-0.85em}\section{Related Work}
\label{sec:relatedwork}

\textbf{Modal Logical Neural Networks.}
MLNNs~\citep{sulc2025mlnn} introduced differentiable Kripke semantics with $\Box/\Diamond$ neurons over discrete worlds.  CMLNNs extend this from graph-based logic (worlds as nodes, accessibility as edges) to fluid logic (worlds on manifolds, accessibility via stochastic flow), while supporting multi-modal reasoning via an SDE library.

\textbf{Neural SDEs.}
Neural SDEs~\citep{li2020scalable} extend Neural ODEs~\citep{chen2018neural} with learned diffusion.  We use them in a fundamentally different way: not for generative modeling or stochastic depth, but to define the \emph{accessibility structure} of modal logics, where diffusion controls the branching of possible worlds.

\textbf{Signal Temporal Logic (STL).}
STL~\citep{maler2004monitoring} provides differentiable $\Box/\Diamond$ over continuous time signals~\citep{donze2010robust}.  The CMLNN temporal operator in isolation is closely related to differentiable STL robustness over SDE rollouts.  CMLNNs extend STL in three directions: (i)~interval truth bounds $[L, U]$ with soundness guarantees; (ii)~\emph{non-temporal} modalities (epistemic, doxastic, deontic) via the multi-SDE architecture, which STL cannot express; (iii)~multi-SDE composition where nested formulas chain operators with different semantics and initializations.

\textbf{Other Related Work.}
Differential Dynamic Logic (d$\mathcal{L}$) provides formal hybrid system verification but requires hand-crafted invariants\footnote{A standard reference is Platzer, \textit{Logical Foundations of Cyber-Physical Systems}, Springer, 2018.}; CMLNNs learn structures from data with probabilistic guarantees.  Neural Lyapunov~\citep{chang2019neural} and barrier methods~\citep{ames2017control} certify single properties but cannot compose temporal invariance with epistemic uncertainty.  Probabilistic model checkers (PRISM, STORM) require finite discrete state spaces.

\vspace{-0.85em}\section{Method: Continuous Modal Logical Neural Networks}
\label{sec:method}

\subsection{From Graph Logic to Fluid Logic}
\label{sec:method:prelim}

In a discrete Kripke model, worlds are graph nodes and accessibility is a fixed edge matrix. CMLNNs lift this to a continuous manifold $\mathcal{W} \subseteq \mathbb{R}^d$: each modality type $i$ is a Neural SDE $dz = f^{(i)}_\theta(z,s)\,ds + \sigma^{(i)}_\theta(z,s)\,dW_s$ whose sample paths define stochastic accessibility. Accessibility becomes reachability under the learned dynamics, and memory scales as $O(|\theta|)$ rather than $O(|W|^2)$. Unlike deterministic ODEs where $\Box\phi = \Diamond\phi$ (quantifier collapse), an SDE fans out into a distribution: $\Box$ acts as a soft worst-case over paths while $\Diamond$ acts as a soft best-case. The diffusion coefficient $\sigma^{(i)}_\theta$ is the learned branching factor that makes these operators genuinely distinct (Theorem~\ref{thm:quantifier_noncollapse}).

\subsection{Continuous Modal Operators}
\label{sec:method:modal_ops}

Atomic formulas return real-valued robustness margins, and all composite formulas propagate robustness in the STL sense (positive means satisfied, negative violated).  For concentration bounds and stable training, we assume robustness is bounded on the relevant domain (either because $\mathcal{W}$ is compact and valuations are Lipschitz, or by explicit clipping): $L_\phi(w), U_\phi(w) \in [-B,B]$ for some $B>0$.  When a $[0,1]$ truth degree is needed for visualization, we report the normalized value $\tilde{v}_\phi(w) = \sigma(L_\phi(w)/\beta)$ using a sigmoid $\sigma$ and scale $\beta$.  Figures and tables are labeled explicitly: ``robustness'' when reporting raw $L$, $U$; ``$\tilde{v}$'' or ``normalized'' when reporting values in $[0,1]$.

The CMLNN maintains an interval $[L_\phi(w),\, U_\phi(w)]$ bracketing the robustness of formula $\phi$ at world $w$. To evaluate modal properties, we draw $N_{\mathrm{mc}}$ sample paths $\{\omega_n\}$ of SDE~$i$ from $w$ over $K$ time-grid points $\{s_k\}$. We first compute the per-path score (a soft minimum of $L_\phi$ along path $\omega_n$), which finds the ``weakest link'' in $\phi$-satisfaction over time:
\begin{equation}
g_n = -\tau_s \log \frac{1}{K}\sum_{k=1}^{K} \exp\!\left(-\frac{L_\phi(\Phi^{(i),s_k}_\theta(w; \omega_n))}{\tau_s}\right)
\label{eq:per_path_score}
\end{equation}
The necessity ($\Box_i$) and possibility ($\Diamond_i$) operators aggregate these scores across paths:
\begin{align}
L_{\Box_i\phi}(w) &= -\tau_\omega \log \tfrac{1}{N_{\mathrm{mc}}} \sum_{n} \exp\!\left(-{g_n}/{\tau_\omega}\right)
\label{eq:stoch_necessity} \\
U_{\Diamond_i\phi}(w) &= \;\;\tau_\omega \log \tfrac{1}{N_{\mathrm{mc}}} \sum_{n} \exp\!\left({h_n}/{\tau_\omega}\right)
\label{eq:stoch_possibility}
\end{align}
where $h_n$ is the soft maximum of $U_\phi$ along $\omega_n$. This makes both operators entropic risk measures at temperature $\tau$: $\Box_i$ uses a soft minimum to robustly identify the worst-case path, while $\Diamond_i$ uses a soft maximum to identify the best-case path. As $\tau\to 0^+$, they recover classical $\forall/\exists$ quantifiers. The key structural property is the soundness gap:
\begin{equation}
L_{\Box_i\phi}(w) \;\leq\; L_\phi(w) + \tau_s \log K + \tau_\omega \log N_{\mathrm{mc}}
\label{eq:softmin_gap}
\end{equation}
This is a T-style inclusion property: because $s{=}0$ is among the softmin arguments, the inner softmin yields $g_n \leq L_\phi(w) + \tau_s\log K$, and the outer aggregation adds $\tau_\omega\log N_{\mathrm{mc}}$; as $\tau_s,\tau_\omega\!\to\!0$, the gap vanishes. Population-level equations appear in Appendix~\ref{sec:appendix:population_ops}.

\subsection{The Multi-Modal SDE Library}
\label{sec:method:multimodal}

The unifying architecture is a library of named SDEs, one per modality:
\begin{equation}
\mathcal{F} = \left\{ \left(i,\; f^{(i)}_\theta,\; \sigma^{(i)}_\theta,\; \texttt{init}^{(i)} \right) \right\}_{i \in I}
\label{eq:sde_library}
\end{equation}
Nested formulas compose SDEs.  For $B_a(\Box_{[0,T]}\text{safe})$: (1)~draw doxastic sample paths from $\hat{x}_a$; (2)~along each, evaluate temporal $\Box$ using the physical SDE; (3)~aggregate via outer softmin.

The temporal operator ($\Box_{[a,b]}\phi$) uses an SDE with drift $f^{(\text{temp})}_\theta$ that learns physical dynamics, initialized at the true state $z(0) = x_{\text{true}}$.  The epistemic operator ($K_a\phi$) uses a conditional SDE $f^{(\text{epist})}_{\theta,a}(z, y_a, s)$ that explores observation-consistent states from $z(0) = x$, serving as a continuous relaxation of classical S5 knowledge.  The doxastic operator ($B_a\phi$) uses the agent's (possibly wrong) internal model $\hat{f}^{(\text{dox})}_{\theta,a}$ initialized at $z(0) = \hat{x}_a \neq x$; Axiom~T fails by construction, since beliefs can be false.  The deontic operator ($O\phi$) uses an SDE trained to stay within a permissible region $\mathcal{P}$; $O\phi \not\to \phi$ because obligations can be violated.

The epistemic/doxastic distinction reduces to a single design choice: SDE initialization at the true state ($K_a$, Axiom~T encouraged) versus the believed state ($B_a$, Axiom~T fails).  Hallucination detection is $B_a(\text{safe}) \land \neg\text{safe}$: the doxastic and temporal SDEs disagree.

\subsection{Axiomatic Properties}
\label{sec:method:axioms}

Classical modal axioms correspond to structural properties of the SDE architecture (full correspondence in Table~\ref{tab:axiom_sde}, Appendix).  These are structural analogues, not formal proofs of classical axiom satisfaction.  Axiom~T (veridicality, $\Box\phi \to \phi$) is encouraged because $s{=}0$ is in the time grid, so $L_\phi(w)$ is always among the softmin arguments: $L_{\Box_i\phi}(w) \leq L_\phi(w) + \tau_s\log K + \tau_\omega \log N_{\mathrm{mc}}$ (Eq.~\ref{eq:softmin_gap}); the gap vanishes as $\tau_s,\tau_\omega \to 0$.  Axiom~D (seriality) holds exactly since any SDE generates at least one path.  Axiom~4 (transitivity) holds approximately via the Markov semigroup property.  Axiom~T fails by construction for the doxastic SDE ($z(0) = \hat{x} \neq x$), which is what enables hallucination detection.

\subsection{Training}
\label{sec:method:learning}

The CMLNN is trained end-to-end:
\begin{equation}
\mathcal{L}_{\text{total}} = \mathcal{L}_{\text{task}} + \beta \mathcal{L}_{\text{contra}} + \gamma \mathcal{L}_{\text{physics}} + \sum_{i} \lambda_i \mathcal{L}_{\text{axiom}}^{(i)}
\label{eq:total_loss}
\end{equation}
where $\mathcal{L}_{\text{contra}}$ penalizes contradictions ($L > U$), $\mathcal{L}_{\text{physics}}$ optionally enforces known physics, and $\mathcal{L}_{\text{axiom}}^{(i)}$ encourages modality-specific axioms.  Gradients flow through the SDE solver via the stochastic adjoint method~\citep{li2020scalable}.  In pure satisfiability mode, $\mathcal{L}_{\text{task}} = 0$ and all supervision comes from logical specifications.

\vspace{-0.85em}\section{Theoretical Analysis}
\label{sec:theory}

\subsection{Soundness}
\label{sec:theory:soundness}

The CMLNN operators are sound with respect to entropic risk-based modal semantics, not classical Kripke semantics.  At finite temperature, the operators are more optimistic than the essential infimum and can miss rare-but-catastrophic paths.

\begin{theorem}[Soundness of CMLNN Risk-Based Bounds]
\label{thm:soundness}
Let a CMLNN be initialized with theory $\Gamma_0$ and SDE library $\mathcal{F}$ with non-empty consistent probability measures. Assume the per-path score is bounded $g_n \in [-B,B]$ (e.g., via bounded/Lipschitz valuations on compact $\mathcal{W}$ or explicit clipping), hence $\hat{L}_\phi \in [-B,B]$.  For the concentration bound, assume $B \leq C\tau_\omega$ for some constant $C$ (e.g., via clipping $g/\tau_\omega$ or rescaling so the exponential range stays bounded):
\textbf{(a)}~The population soft operators satisfy $L^{\mathrm{pop}}_\phi(w) \leq \mathcal{P}^{\mathrm{risk}}_\tau(S_{\phi,w})$ for the entropic risk measure $\mathcal{P}^{\mathrm{risk}}_\tau$ (deterministic, no MC error).
\textbf{(b)}~Under $B \leq C\tau_\omega$, the MC estimator concentrates: $\mathbb{P}( | \hat{L}_\phi - L^{\mathrm{pop}}_\phi | \geq \epsilon ) \leq 2\exp\!\big(-c(C)\,N_{\mathrm{mc}}\epsilon^2/\tau_\omega^2\big)$ for some $c(C) > 0$ depending on $C$ (an explicit expression is given in Appendix~\ref{sec:appendix:proofs}).
\textbf{(c)}~As $\tau \to 0^+$, the softmin/softmax converge to $\operatorname{ess\,inf}$/$\operatorname{ess\,sup}$; classical Kripke semantics (Boolean truth at each world) follows only under an explicit \emph{Booleanization} assumption on atomic valuations (e.g., sign-thresholding or $\{0,1\}$-valued atoms).  In the fixed-$B$ regime (without $B \leq C\tau_\omega$), sample complexity diverges as $\tau_\omega \to 0$ for fixed $B$, reflecting the $e^{\Theta(B/\tau_\omega)}$ blow-up in the concentration bound.
\end{theorem}

Theorem~\ref{thm:soundness} ensures the network's outputs are mathematically sound lower bounds on system safety up to a chosen risk threshold $\tau$, with Monte Carlo concentration. Crucially, the soundness is strictly w.r.t.\ entropic risk measures, not classical Kripke semantics. At finite $\tau$, risk measures are more optimistic than the essential infimum and may ignore rare catastrophic paths, acting as a well-posed optimization target rather than an airtight verification guarantee. Full proof in Appendix~\ref{sec:appendix:proofs}.

\subsection{Quantifier Non-Collapse}
\label{sec:theory:quantifier}

\begin{theorem}[Quantifier Non-Collapse]
\label{thm:quantifier_noncollapse}
Let $\sigma^{(i)}_\theta$ be non-degenerate on a set of positive measure.  Then there exists $\phi, w$ such that $L_{\Box_i\phi}(w) < U_{\Diamond_i\phi}(w)$.
\end{theorem}

In deterministic systems, ``all futures'' ($\Box$) and ``some future'' ($\Diamond$) collapse to the exact same statement because there is only one trajectory. Theorem~\ref{thm:quantifier_noncollapse} shows that stochastic diffusion prevents this quantifier collapse by branching paths into distributions, making necessity and possibility genuinely distinct concepts.

\begin{proof}[Sketch]
Non-degenerate diffusion implies positive variance at some $s^*$.  Choose a smooth Lipschitz predicate varying across the support.  \textit{Example:} $\mathcal{W} = \mathbb{R}$, $\phi(z)=\tanh(z)$, $w=0$, $dz=\sigma\,dW_s$.  Then with positive probability $\phi(\Phi^{s^*}(0))$ is close to $1$ and with positive probability it is close to $-1$, so the soft worst-case and soft best-case aggregations differ, yielding $L_{\Box_i\phi}(w) < U_{\Diamond_i\phi}(w)$.  For deterministic ODEs ($\sigma \equiv 0$), the flow is a single trajectory and the two operators coincide.
\end{proof}

\subsection{Complexity}
\label{sec:theory:complexity}

One modal operator evaluation costs $O(N_{\mathrm{mc}} \cdot K \cdot C_f)$ where $C_f = O(d \cdot |\theta_i|)$; sample paths are embarrassingly parallel.  For nesting depth $D$, worst-case cost is $O(N_{\mathrm{mc}}^D \cdot K^D \cdot C_f)$.  Our experiments use $D \leq 2$.  Additional theorems (convergence, universal approximation, sample complexity, differentiability) are in Appendix~\ref{sec:appendix:additional_theorems}.

\vspace{-0.85em}\section{Experiments}
\label{sec:experiments}

We evaluate CMLNNs on three case studies spanning epistemic, doxastic, temporal, and deontic modal logics to validate the framework's multi-modal generality on diverse domains.  An overview is provided in Appendix~\ref{sec:appendix:experiments}.

\subsection{Case Study 1: Epistemic + Doxastic Logic, Swarm Hallucination Detection}
\label{sec:exp:swarm}

We simulate a custom 2D environment where a 5-rover swarm navigates toward a target while avoiding a hazardous spatial region (``chasm''). This allows defining precise multimodal logic fields over continuous space without discrete constraints. Rover~3 suffers a sensor failure, causing its internal model to hallucinate a fake chasm elsewhere, while believing the true chasm is safe. The key formula:
\begin{equation}
B_{\text{rover3}}(\Box_{[0,T]}\text{safe}) \;\land\; K_{\text{swarm}}(\Diamond_{[0,T]}\text{collision})
\label{eq:swarm_hallucination}
\end{equation}
requires two distinct SDEs: doxastic (from Rover~3's believed state) and epistemic (conditioned on healthy sensors). Standard single-SDE or pure STL formulations cannot naturally express this conjunction without defining separate evaluation models.

Three SDEs share state space $\mathcal{W} \subseteq \mathbb{R}^{4 \times 5}$: temporal (true physics, $z(0) = x_{\text{true}}$), epistemic (swarm knowledge, $z(0) = x_{\text{true}}$, conditioned on healthy sensors), and doxastic (Rover~3's internal model, $z(0) = \hat{x}_3$, using outdated dynamics).  The hallucination flag triggers when both $L_{B_3}(\Box\,\text{safe}) > 0.8$ and $U_{K_{\text{swarm}}}(\Diamond\,\text{collision}) > 0.3$.

\begin{figure}[!t]
\vspace{-1em}
\centering
\includegraphics[width=0.8\linewidth]{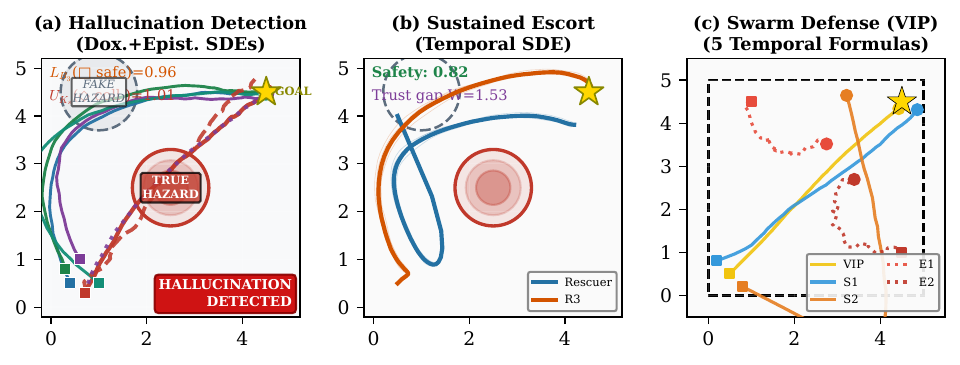}
\vspace{-0.5em}
\caption{Case Study~1: three multi-agent scenarios.
(a)~Hallucination detection (epistemic + doxastic SDEs): Rover~3 (faulty sensor) believes safety where danger exists.  Solid lines show true trajectories; dashed red is the doxastic path; dotted purple is the epistemic path.  The hallucination flag triggers when $L_{B_3}(\Box\,\text{safe}){>}0.8$ and $U_{K_{\text{sw}}}(\Diamond\,\text{coll.}){>}0.3$.
(b)~Sustained escort (temporal SDE): a rescuer (blue) learns to shield Rover~3 (orange) around the true chasm toward the goal (star); final actual safety $0.82$, trust gap $W{=}1.53$.
(c)~Swarm defense (5 temporal formulas): the VIP (yellow) is escorted by scouts S1/S2 while evading enemies E1/E2; five logic objectives are jointly optimised. Additional detail in Appendix~\ref{sec:appendix:experiments}.}
\label{fig:cs1_hallucination}
\vspace{-1.5em}
\end{figure}

\paragraph{Results.}
Hallucination is detected at epoch~1 and persists throughout training (Fig.~\ref{fig:cs1_hallucination}).  During joint training ($N_{\mathrm{mc}}{=}24$), the doxastic bound $L_{B_3}(\Box\,\text{safe})$ consistently exceeds the 0.8 detection threshold (peaking at 1.04), while the epistemic bound $U_{K_{\text{swarm}}}(\Diamond\,\text{collision})$ exceeds 0.3, confirming that Rover~3 believes all futures are safe while the swarm knows collision is possible.  At final evaluation with higher MC precision ($N_{\mathrm{mc}}{=}256$), the more conservative estimates give $L_{B_3} = 0.73$ and $U_{K_{\text{swarm}}} = 1.25$; the doxastic bound falls slightly below the 0.8 threshold, consistent with the MC concentration bound (Prop.~\ref{prop:sample_complexity}).  The Wasserstein belief--reality gap stabilises at $W = 0.39$, quantifying the doxastic--epistemic divergence.  The spatial modal truth field (Fig.~\ref{fig:cs1_modal_field} in appendix) shows that detection arises specifically when Rover~3's initial position places it on a trajectory through the true chasm, i.e.\ exactly where the two SDEs spatially disagree.  Two sub-experiments validate additional multi-agent formulas: (B)~a 2-rover sustained escort scenario where the rescuer learns to shield Rover~3 (final safety $0.82$, trust gap $1.53$), and (C)~a 5-agent swarm defense (VIP protection) with five jointly optimised temporal formulas (in\_bounds: $0.16$, safe: $0.12$, follow: $-0.03$, attack: $-0.20$, explore: $-0.47$).

\subsection{Case Study 2: Temporal Logic, Lorenz Chaotic System (LINNs)}
\label{sec:exp:lorenz}

The Lorenz-63 system is a classic model of chaotic fluid convection, chosen here because its distinct two-lobe butterfly attractor is highly sensitive to initial conditions. This chaotic nature makes it notoriously difficult for deterministic neural models to capture its global structure once trajectories diverge. To address this, we introduce Logic-Informed Neural Networks (LINNs), analogous to PINNs but enforcing modal logical trajectory properties instead of local differential equations:
\begin{equation}
\Box_{[0,T]}\text{(bounded)} \;\land\; \Diamond_{[0,T]}\text{(visits\_lobe)}
\label{eq:lorenz_formula}
\end{equation}
This requires $\Box \neq \Diamond$ on the Lorenz-63 attractor: all trajectories must stay bounded, yet some must visit the secondary lobe.  A deterministic ODE cannot satisfy this (quantifier collapse); the SDE's diffusion provides genuine branching.  Eight models are compared (Neural ODE, Pure MSE SDE, PINN, PINN+noise, Ensemble ODE, ODE+LINN, SDE+LINN, PINN+LINN).

\begin{figure}[!t]
\vspace{-1em}
\centering
\includegraphics[width=0.8\linewidth]{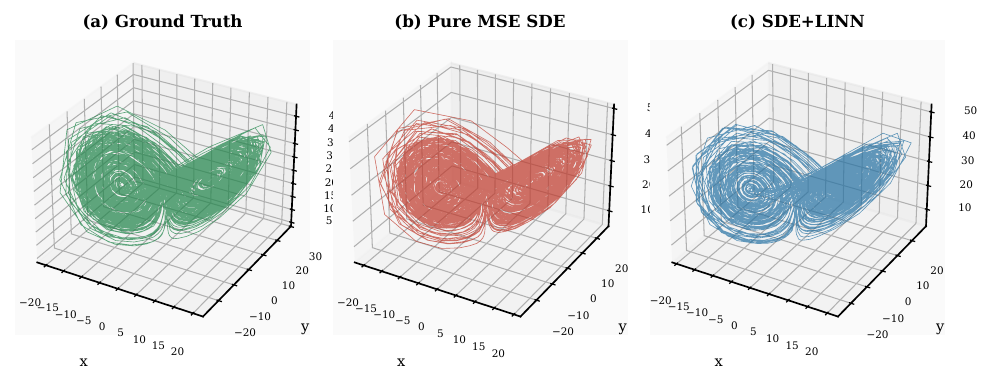}
\vspace{-0.5em}
\caption{Case Study~2: 3D attractor comparison on Lorenz-63. Ground truth shows the canonical two-lobe butterfly; the Pure MSE SDE collapses to one lobe; SDE+LINN recovers both lobes via $\Box(\text{bounded})\land\Diamond(\text{visits\_lobe})$.}
\label{fig:cs2_attractor}
\vspace{-1.5em}
\end{figure}

\begin{figure}[!t]
\vspace{-1em}
\centering
\begin{minipage}{0.48\linewidth}
\centering
\includegraphics[width=0.85\linewidth]{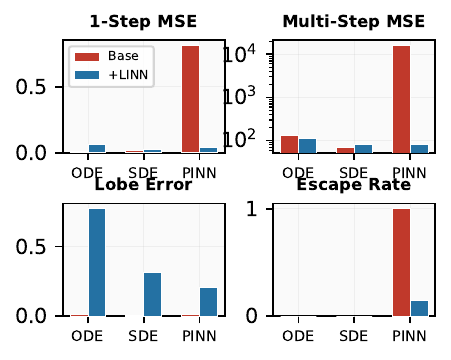}
\end{minipage}\hfill
\begin{minipage}{0.48\linewidth}
\centering
\includegraphics[width=0.85\linewidth]{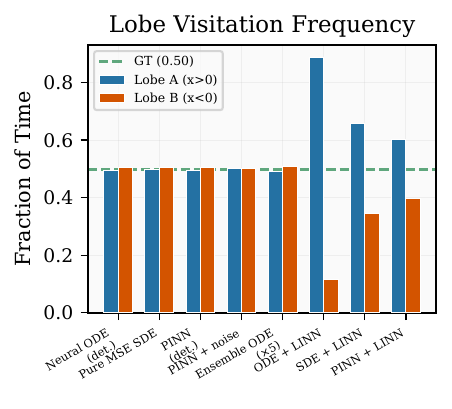}
\end{minipage}
\vspace{-0.5em}
\caption{Case Study~2: LINN as a logic regularizer. Left: Adding LINN to any base model reduces multi-step MSE and escape rate. Right: Lobe visitation frequency. Baseline models appear to have a balanced average Lobe Error, but suffer from structural collapse (individual trajectories rarely switch lobes). LINN forces explicit exploration of both lobes, recovering the global geometry at the cost of slightly perturbing the natural switching frequency.}
\label{fig:cs2_regularizer}
\vspace{-1.5em}
\end{figure}

\paragraph{Results.}
Table~\ref{tab:cs2_results} reports all metrics. Ground truth modal values are computed via 256-path Monte Carlo on the stochastic Lorenz system in the same bounded robustness scale used by our operators (Sec.~\ref{sec:method:modal_ops}; e.g., $\Box(\text{bounded}) = 0.69$, $\Delta_Q = 0.25$). Deterministic models (Neural ODE, PINN) exhibit quantifier collapse ($\Box = \Diamond$ by construction, marked~$\dagger$). The purely data-driven ODE/SDE models minimize trajectory MSE, which naturally keeps them bounded. In contrast, PINNs rely solely on minimizing the local residual. Because Lorenz-63 is highly chaotic, small numerical errors compound rapidly, causing the unconstrained PINN to diverge entirely over long horizons (100\% escape rate). This explicitly serves as an ablation: the comparison is not to claim SOTA physics forecasting, but to demonstrate that correct local dynamics do not ensure global structural stability.

While the baseline Neural ODE and SDE do not escape, they structurally collapse, failing to capture the butterfly geometry. Their Lobe Error (population-level balance of time spent in the left vs.\ right lobe) appears artificially low because different initial conditions collapse into different single lobes, averaging to 50/50 despite individual trajectories never switching. The LINN explicitly forces exploration of both lobes ($\Diamond\text{visits\_lobe}$). This perfectly pinpoints what the symbolic component buys us: LINN fundamentally solves the structural collapse. SDE+LINN achieves the best $\Box$ MAE (0.088) and $\Diamond$ MAE (0.070). Visually (Fig.~\ref{fig:cs2_attractor}), it is the only model recovering the true two-lobe geometry. Adding LINN to any base model acts as a powerful regularizer, reducing divergence (Fig.~\ref{fig:cs2_regularizer}).

\begin{table}[t]
\centering
\caption{Case Study~2 results (Lorenz-63).  $\dagger$: quantifier collapse ($\Box = \Diamond$).}
\label{tab:cs2_results}
\resizebox{\linewidth}{!}{
\begin{tabular}{@{}lrrrrrrr@{}}
\toprule
\textbf{Model} & \textbf{1-Step} & \textbf{Multi-S} & $\Box$ \textbf{MAE} & $\Diamond$ \textbf{MAE} & $\Delta_Q$ \textbf{MAE} & \textbf{Lobe Err} & \textbf{Esc.\%} \\ \midrule
Neural ODE & 0.013 & 132.2 & 0.133$^\dagger$ & 0.063$^\dagger$ & 0.124$^\dagger$ & 0.012 & 0 \\
Pure MSE SDE & 0.015 & 68.7 & 0.112 & 0.074 & 0.121 & 0.006 & 0 \\
PINN (det.) & 0.814 & 16084 & 0.400$^\dagger$ & 0.478$^\dagger$ & 0.124$^\dagger$ & 0.010 & 100 \\
PINN + noise & 0.816 & 12523 & 0.400 & 0.477 & 0.123 & 0.001 & 100 \\
Ensemble ODE & 0.013 & 134.4 & 0.119 & 0.071 & 0.119 & 0.016 & 0 \\
ODE + LINN & 0.064 & 111.3 & 0.091 & 0.079 & 0.130 & 0.771 & 0 \\
\textbf{SDE + LINN} & \textbf{0.023} & \textbf{79.6} & \textbf{0.088} & \textbf{0.070} & \textbf{0.109} & 0.313 & \textbf{0} \\
PINN + LINN & 0.037 & 83.6 & 0.204 & 0.121 & 0.171 & 0.203 & 15 \\
\bottomrule
\end{tabular}}
\end{table}

\subsection{Case Study 3: Deontic Logic, Safe Confinement Dynamics}
\label{sec:exp:deontic}

We simulate a Tokamak-inspired 2D arena (poloidal cross-section) where a particle must remain confined within a specific radius. This tests whether safe control can be synthesized entirely from formal logical specifications rather than hand-crafted rewards. Two SDEs model distinct dynamics: the temporal SDE represents unconstrained physics (drifting outward and exiting), while the deontic SDE is trained to maximize $L_{\Box\text{safe}}$.

The environment's ``physics'' is embedded as a fixed Grad--Shafranov-inspired base drift, applying a toroidal swirl and radial pressure gradient that inherently pushes particles outward. The deontic SDE takes this physics as its base and learns an additive restoring force (via an MLP) that counteracts the outward pressure. By directly using $O(\Box_{[0,T]}\text{safe})$ as the training objective, the deontic SDE is effectively trained as a LINN, synthesizing safe dynamics entirely from formal constraints.

\begin{figure}[!t]
\vspace{-1em}
\centering
\includegraphics[width=0.8\linewidth]{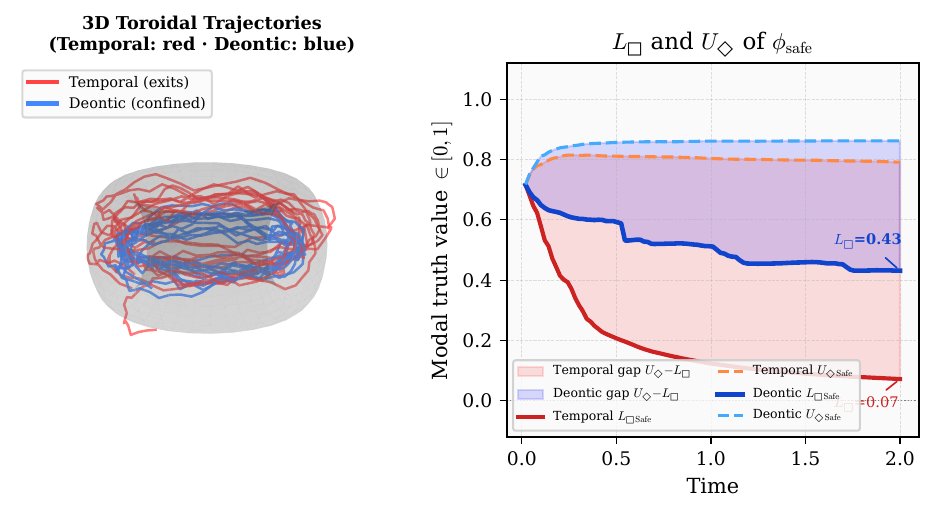}
\vspace{-0.5em}
\caption{Case Study~3: Deontic safe confinement.
Left: 3D toroidal trajectories (poloidal cross-section wrapped onto torus geometry).
Red paths show the temporal SDE (unconstrained physics, drifting outward and exiting the vessel).
Blue paths show the deontic SDE (trained on $O(\Box_{[0,T]}\text{safe})$, staying confined with 0\% exit).
The confinement structure emerges from the logical objective alone, without a hand-crafted control law.
Right: $L_\Box$ and $U_\Diamond$ of $\phi_\text{safe}$ over time.
The deontic SDE (blue) maintains $L_{\Box\text{safe}} = 0.456$ (6.1$\times$ above temporal 0.075); the shaded quantifier gap $U_\Diamond - L_\Box > 0$ confirms non-collapsed quantifiers in both SDEs.
Drift magnitude heatmaps showing the learned restoring force appear in Appendix~\ref{sec:appendix:experiments}.}
\label{fig:cs6_deontic}
\vspace{-1.5em}
\end{figure}

\paragraph{Results.}
Table~\ref{tab:cs6_results} summarises the key metrics.  The temporal SDE (natural dynamics) has $L_{\Box\text{safe}} = 0.075$ with 61.7\% of paths exiting the vessel; the deontic SDE (trained to maximise $O(\Box\text{safe})$) achieves $L_{\Box\text{safe}} = 0.456$ with 0\% exit probability, a $6.1\times$ improvement (Fig.~\ref{fig:cs6_deontic}).  The quantifier gap ($U_{\Diamond} - L_{\Box}$) is 0.72 for temporal vs.\ 0.40 for deontic, confirming non-collapsed quantifiers in both SDEs.  Physical consistency tests confirm the learned deontic drift is inward in the scrape-off layer (100\% of samples), with magnitude increasing with radius (correlation 0.84) and generalising to mid-radius regions ($r \in [0.5, 0.7]$) not explicitly trained on (visible as the bright boundary ring in the drift heatmaps; see Fig.~\ref{fig:cs6_heatmap} in Appendix~\ref{sec:appendix:experiments}).  A trivial radial baseline ($f = -k\hat{r}$ for $r > 0.6$) achieves only $L_{\Box} = 0.148$ and 24.2\% exit rate, confirming the learned drift captures non-trivial spatial structure.

\begin{table}[t]
\centering
\caption{Case Study~3 results (deontic confinement).}
\label{tab:cs6_results}
\small
\begin{tabular}{@{}lcccc@{}}
\toprule
\textbf{SDE} & $L_{\Box\text{safe}}$ & $U_{\Diamond\text{safe}}$ & \textbf{Exit \%} & $\Delta_Q$ \\ \midrule
Temporal & 0.075 & 0.797 & 61.7 & 0.722 \\
\textbf{Deontic} & \textbf{0.456} & \textbf{0.857} & \textbf{0.0} & 0.401 \\
Trivial baseline & 0.148 & --- & 24.2 & --- \\
\bottomrule
\end{tabular}
\end{table}

\vspace{-0.85em}\section{Conclusion}
\label{sec:conclusion}

We introduced Continuous Modal Logical Neural Networks (CMLNNs) and the Fluid Logic paradigm: each modal operator type is backed by a dedicated Neural SDE, with the diffusion coefficient governing the semantic distinction between necessity and possibility.  The framework provides sound bounds with respect to entropic risk-based semantics (Theorem~\ref{thm:soundness}), resolves quantifier collapse (Theorem~\ref{thm:quantifier_noncollapse}), and identifies structural analogues between SDE properties and modal axioms.  A key instantiation of this framework is Logic-Informed Neural Networks (LINNs), which allow modal logical formulas to act as training objectives, enabling the recovery of structural properties that local dynamics alone cannot enforce. Furthermore, the epistemic/doxastic distinction is captured by a single design choice (SDE initialization), and the deontic SDE provides a logic-structured specification layer for safe control.

In future, promising directions include extending Fluid Logic to other modalities, such as strategic or coalition logics, for reasoning about complex multi-agent systems. Future research may also focus on improving the $O(N_{\mathrm{mc}}^D)$ scalability of deep formula nesting and developing variance reduction techniques to tighten the probabilistic Monte Carlo bounds.

\label{lastpage}
\bibliography{references}

\clearpage
\appendix

\vspace{-0.85em}\section{Operator Definitions (Population Form)}
\label{sec:appendix:population_ops}

\subsection{Necessity ($\Box$) vs.\ Possibility ($\Diamond$), Visual Proof}

\begin{figure}[h]
\centering
\begin{tikzpicture}[font=\sffamily\small, >=Stealth,
  pathstyle/.style={thick, rounded corners=2mm},
  worldpt/.style={circle, fill, inner sep=2pt},
  valbox/.style={rectangle, rounded corners=1pt, draw=gray!60, fill=gray!8,
                 minimum width=1.0cm, minimum height=0.38cm, align=center, font=\scriptsize},
]
\node[font=\bfseries\sffamily, anchor=west] at (-0.3, 2.4) {$\Box_i\phi(w)$: soft \textbf{worst-case} over paths};

\node[worldpt, fill=teal!80] (w0L) at (0,0) {};
\node[below=2pt of w0L, font=\scriptsize] {$w$};

\draw[pathstyle, color=red!80,   opacity=0.85] (w0L.east) .. controls (0.8, 0.9) .. (2.6, 1.2)
      node[right, font=\scriptsize, text=red!80] {$g_1\!=\!0.2$};
\draw[pathstyle, color=orange!80, opacity=0.85] (w0L.east) .. controls (0.8, 0.3) .. (2.6, 0.5)
      node[right, font=\scriptsize, text=orange!80] {$g_2\!=\!0.5$};
\draw[pathstyle, color=green!60!black, opacity=0.85] (w0L.east) .. controls (0.8,-0.3) .. (2.6,-0.4)
      node[right, font=\scriptsize, text=green!60!black] {$g_3\!=\!0.8$};
\draw[pathstyle, color=blue!70, opacity=0.85] (w0L.east) .. controls (0.8,-0.9) .. (2.6,-1.1)
      node[right, font=\scriptsize, text=blue!70] {$g_4\!=\!0.9$};

\draw[pathstyle, color=red!80, line width=2.5pt, opacity=0.4]
      (w0L.east) .. controls (0.8, 0.9) .. (2.6, 1.2);
\node[draw=red!70, fill=red!10, rounded corners=2pt, font=\scriptsize, inner sep=3pt]
     at (1.3, 1.55) {worst-case path};

\node[valbox, fill=red!10, draw=red!60] at (1.3, -1.65)
     {$L_{\Box\phi}(w) \approx g_1 = 0.2$};

\begin{scope}[shift={(6.2,0)}]
\node[font=\bfseries\sffamily, anchor=west] at (-0.3, 2.4) {$\Diamond_i\phi(w)$: soft \textbf{best-case} over paths};

\node[worldpt, fill=teal!80] (w0R) at (0,0) {};
\node[below=2pt of w0R, font=\scriptsize] {$w$};

\draw[pathstyle, color=red!80,   opacity=0.85] (w0R.east) .. controls (0.8, 0.9) .. (2.6, 1.2)
      node[right, font=\scriptsize, text=red!80] {$h_1\!=\!0.2$};
\draw[pathstyle, color=orange!80, opacity=0.85] (w0R.east) .. controls (0.8, 0.3) .. (2.6, 0.5)
      node[right, font=\scriptsize, text=orange!80] {$h_2\!=\!0.5$};
\draw[pathstyle, color=green!60!black, opacity=0.85] (w0R.east) .. controls (0.8,-0.3) .. (2.6,-0.4)
      node[right, font=\scriptsize, text=green!60!black] {$h_3\!=\!0.8$};
\draw[pathstyle, color=blue!70, opacity=0.85] (w0R.east) .. controls (0.8,-0.9) .. (2.6,-1.1)
      node[right, font=\scriptsize, text=blue!70] {$h_4\!=\!0.9$};

\draw[pathstyle, color=blue!70, line width=2.5pt, opacity=0.4]
      (w0R.east) .. controls (0.8,-0.9) .. (2.6,-1.1);
\node[draw=blue!70, fill=blue!10, rounded corners=2pt, font=\scriptsize, inner sep=3pt]
     at (1.3, -1.65) {};
\node[draw=blue!70, fill=blue!10, rounded corners=2pt, font=\scriptsize, inner sep=3pt]
     at (1.3, 1.55) {best-case path};

\node[valbox, fill=blue!10, draw=blue!60] at (1.3, -1.65)
     {$U_{\Diamond\phi}(w) \approx h_4 = 0.9$};
\end{scope}

\draw[<->, thick, gray!60] (4.0, -1.65) -- (5.7, -1.65)
     node[midway, below, font=\scriptsize, text=gray] {gap $U_\Diamond - L_\Box = 0.7 > 0$};

\end{tikzpicture}
\caption{Quantifier non-collapse visualised.  From the same starting world $w$, four SDE sample paths fan out with different per-path truth scores $g_n$ (soft minima of $L_\phi$ along the path).
Left ($\Box$): necessity aggregates via soft min, so the worst-case path (red, $g=0.2$) dominates.
Right ($\Diamond$): possibility aggregates via soft max, so the best-case path (blue, $h=0.9$) dominates.
Because the SDE produces a distribution of paths, $L_\Box < U_\Diamond$ (the gap is $0.7$ here).  For a deterministic ODE, all four paths would coincide and $L_\Box = U_\Diamond$ (quantifier collapse).}
\label{fig:box_diamond_visual}
\end{figure}
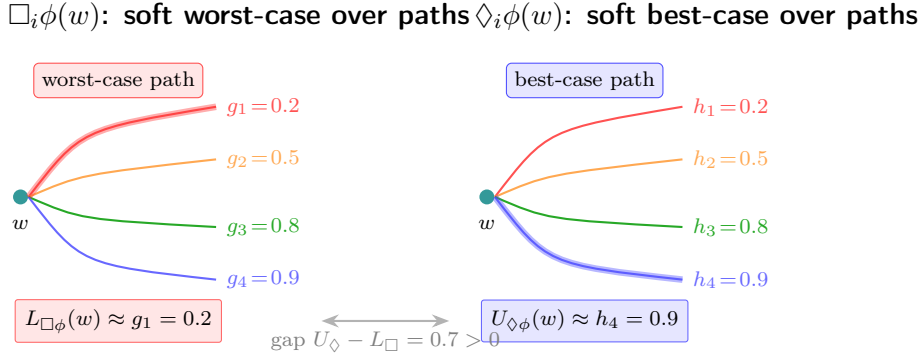

\subsection{Population-Level Operators}

\begin{definition}[Population Entropic Risk Functional]
\label{def:entropic_risk}
For a random variable $X(\omega)$ on $(\Omega, \mathcal{F}, \mathbb{P})$, the population softmin at temperature $\tau > 0$ is:
$$\mathrm{softmin}_{\tau,\,\omega \sim \mathbb{P}}(X(\omega)) := -\tau \log \mathbb{E}_{\mathbb{P}}\!\left[\exp\!\left(-\frac{X(\omega)}{\tau}\right)\right]$$
Monotone, satisfies $\operatorname{ess\,inf} X \leq \mathrm{softmin}_\tau(X) \leq \mathbb{E}[X]$, converges to $\operatorname{ess\,inf} X$ as $\tau \to 0^+$.
\end{definition}

\begin{definition}[Population Necessity and Possibility Operators]
$$L^{\mathrm{pop}}_{\Box_i\phi}(w) = \mathrm{softmin}_{\tau_\omega,\,\omega \sim \mathbb{P}} \!\left( -\tau_s \log \int_0^S \exp\!\left( -\frac{L_\phi(\Phi^{(i),s}_\theta(w; \omega))}{\tau_s} \right) ds \right)$$
The MC operators (Eqs.~\ref{eq:per_path_score}--\ref{eq:stoch_possibility}) replace expectations with empirical means over $N_{\mathrm{mc}}$ paths and integrals with sums over $K$ quadrature points.
\end{definition}

\vspace{-0.85em}\section{Additional Theorem Statements}
\label{sec:appendix:additional_theorems}

\begin{theorem}[Convergence]
\label{thm:convergence}
For bounded drift/diffusion, compact $\mathcal{W}$, Lipschitz valuations, and finite formula set, the CMLNN inference algorithm converges to a fixed point a.s., provided $\tau_s, \tau_\omega > 0$.
\end{theorem}

\begin{theorem}[Universal Approximation of Transition Laws]
\label{thm:expressivity}
For compact $\mathcal{W}$ and any target transition law admitting a Lipschitz It\^o diffusion representation, $\exists$ a Neural SDE with $\sup_w W_1(\text{Law}(\Phi^S_\theta(w; \cdot)), \mathcal{R}(w, \cdot)) < \epsilon$.
\end{theorem}

\begin{theorem}[Wasserstein Universal Approximation]
\label{thm:wasserstein_approx}
Under the same conditions with bounded $p$-th moments, the approximation strengthens to $W_p$.
\end{theorem}

\begin{proposition}[Sample Complexity]
\label{prop:sample_complexity}
Under $B \leq C\tau_\omega$, $\mathbb{P}( | \hat{L}_{\Box_i\phi} - L^{\mathrm{pop}}_{\Box_i\phi} | \geq \epsilon ) \leq 2 \exp\!\big(-c(C)\, N_{\mathrm{mc}} \epsilon^2 / \tau_\omega^2\big)$ for some $c(C) > 0$ (explicit in Appendix~\ref{sec:appendix:proofs}).  For error $\leq \epsilon$ with confidence $1-\delta$: $N_{\mathrm{mc}} \geq \frac{\tau_\omega^2}{c(C)\,\epsilon^2} \log(2/\delta)$.  Without this assumption, the bound includes an $e^{\Theta(B/\tau_\omega)}$ factor (see proof).
\end{proposition}

\begin{theorem}[Differentiability]
\label{thm:diff_accessibility}
Under Lipschitz conditions with bounded derivatives, the stochastic accessibility $\mathcal{R}^{(i)}_\theta(w, w')$ is continuously differentiable in $\theta$ and $w$, computable via the stochastic adjoint.
\end{theorem}

\vspace{-0.85em}\section{Complete Proofs}
\label{sec:appendix:proofs}

\subsection{Proof of Theorem~\ref{thm:soundness} (Soundness)}

\begin{proof}
\textbf{Part (a):} By structural induction on $\phi$.  Base case: $L_p(w) \leq \mathcal{P}(S_{p,w})$ by assumption.  Inductive case ($\Box_i\phi$): by hypothesis $L_\phi(z) \leq \mathcal{P}^{\mathrm{risk}}_\tau(S_{\phi,z})$ pointwise.  Since population softmin (Definition~\ref{def:entropic_risk}) is monotone:
$$L^{\mathrm{pop}}_{\Box_i\phi}(w) \leq \mathrm{softmin}_{\tau_\omega,\omega}(\mathrm{softmin}_{\tau_s,s}(\mathcal{P}^{\mathrm{risk}}_\tau(S_{\phi,\Phi^s(w;\omega)}))) = \mathcal{P}^{\mathrm{risk}}_\tau(S_{\Box_i\phi,w})$$

\textbf{Part (b):} Define $X_n = \exp(-g(\omega_n)/\tau_\omega)$; the $\{X_n\}$ are i.i.d.\ in $[e^{-B/\tau_\omega}, e^{B/\tau_\omega}]$.  The range $R_X = e^{B/\tau_\omega} - e^{-B/\tau_\omega}$ blows up as $\tau_\omega \to 0$ unless $B = O(\tau_\omega)$.  Under $B \leq C\tau_\omega$, $R_X \leq 2(e^C - e^{-C})$ is bounded; Hoeffding on $\bar{X}$ and propagation through $h(x) = -\tau_\omega\log x$ (Lipschitz constant $\tau_\omega/\xi \leq \tau_\omega e^{B/\tau_\omega} \leq \tau_\omega e^C$) yields $\mathbb{P}(|\hat{L}-L^{\mathrm{pop}}| \ge \epsilon)\le 2\exp(-c(C)\,N_{\mathrm{mc}}\epsilon^2/\tau_\omega^2)$.  Explicitly, $c(C) = 2/\big((e^C-e^{-C})^2\,e^{2C}\big)$ from the Hoeffding range $R_X = e^C - e^{-C}$ and Lipschitz constant $\tau_\omega/\xi \leq \tau_\omega e^C$.  Without $B \leq C\tau_\omega$, the bound acquires an $e^{\Theta(B/\tau_\omega)}$ factor from $R_X^2$ and the Lipschitz constant.

\textbf{Part (c):} As $\tau \to 0$, $\mathrm{softmin}_\tau(X) \to \operatorname{ess\,inf} X$ by monotone convergence.

The operators do not deterministically bound classical Kripke semantics at finite $\tau$.  Since $\mathrm{softmin}_\tau(X) \geq \operatorname{ess\,inf} X$, the CMLNN lower bound may lie above the classical truth value.
\end{proof}

\subsection{Proof of Theorem~\ref{thm:quantifier_noncollapse} (Non-Collapse)}

\begin{proof}
Non-degenerate $\sigma^{(i)}_\theta$ implies $\exists\, w^*, s^*$ such that $\text{Var}(\Phi^{(i),s^*}_\theta(w^*; \cdot)) > 0$ in some coordinate.  Let $m = \mathbb{E}[\Phi^{s^*}(w^*)]$ and choose a smooth Lipschitz predicate, e.g.\ $\phi(z)=\tanh(z_1-m_1)$.  Then some paths yield $\phi \approx 1$ and others yield $\phi \approx -1$, so the soft best-case and soft worst-case aggregations differ, giving $L_{\Box_i\phi}(w^*) < U_{\Diamond_i\phi}(w^*)$.  For $\sigma \equiv 0$, the flow is deterministic: softmin = softmax, so quantifiers collapse.
\end{proof}

\subsection{Proof of Proposition~\ref{prop:sample_complexity} (Sample Complexity)}

\begin{proof}
Define $X_n = \exp(-g(\omega_n)/\tau_\omega)$; the $\{X_n\}$ are i.i.d.\ in $[e^{-B/\tau_\omega}, e^{B/\tau_\omega}]$.  Under $B \leq C\tau_\omega$, the range $R_X$ and Lipschitz constant $\tau_\omega/\xi$ stay bounded; Hoeffding plus propagation through $h(x) = -\tau_\omega\log x$ yields $\mathbb{P}(|\hat{L}-L^{\mathrm{pop}}| \ge \epsilon)\le 2\exp(-c(C)\,N_{\mathrm{mc}}\epsilon^2/\tau_\omega^2)$.  Setting this $\leq \delta$ gives $N_{\mathrm{mc}} \ge \frac{\tau_\omega^2}{c(C)\,\epsilon^2}\log(2/\delta)$.  Without the assumption, $R_X \sim e^{B/\tau_\omega}$ and the exponent acquires $e^{-4B/\tau_\omega}$.
\end{proof}

\vspace{-0.85em}\section{Extended Method Details}
\label{sec:appendix:method}

\subsection{Modal Axiom Correspondence}

\begin{table}[h]
\centering
\caption{Modal axioms and SDE structural counterparts.}
\label{tab:axiom_sde}
\small
\begin{tabular}{@{}llll@{}}
\toprule
\textbf{Axiom} & \textbf{Property} & \textbf{SDE Mechanism} & \textbf{Status} \\ \midrule
\textbf{T}: $\Box\phi \to \phi$ & Veridicality & $z(0) = x_{\text{true}}$ & Encouraged (gap) \\
\textbf{D}: $\Box\phi \to \Diamond\phi$ & Seriality & $\geq 1$ path & Free \\
\textbf{4}: $\Box\phi \to \Box\Box\phi$ & Introspection & Semigroup & Approx. \\
$\neg$\textbf{T} & Fallibility & $z(0) = \hat{x} \neq x$ & Enforced \\
\bottomrule
\end{tabular}
\end{table}

\subsection{Stochastic Trajectory-Based Accessibility (Derived Diagnostic)}

The accessibility of $w'$ from $w$ under modality $i$ is a \emph{derived} quantity (not used in operator definitions):
$$\mathcal{R}^{(i)}_\theta(w, w') = \mathbb{E}_\omega\left[\sup_{s \in [0, S]} \exp\left( -\frac{\| w' - \Phi^{(i),s}_\theta(w; \omega) \|^2}{2\sigma_b^2} \right)\right]$$

\subsection{Dynamic Logic}
Action $\alpha$ corresponds to SDE $f^{(\alpha)}_\theta$; sequential composition $[\alpha;\beta]\phi = [\alpha][\beta]\phi$ chains SDE solutions.

\subsection{Wasserstein Distance as Modal Metric}
The modal Wasserstein distance $\mathcal{W}_p^{(i,j)}(w, s) = W_p(\text{Law}(\Phi^{(i),s}_\theta(w; \cdot)), \text{Law}(\Phi^{(j),s}_\theta(w; \cdot)))$ enables graded hallucination quantification and improved diffusion learning.

\subsection{Adversarial Contradiction Mining}
We actively search for corner-case contradictions by gradient ascent on $c(w) = \max(0, L_\phi(w) - U_\phi(w))^2$ over candidate worlds, then add high-contradiction worlds to the training batch.

\vspace{-0.85em}\section{Extended Experimental Details}
\label{sec:appendix:experiments}

Full experimental setup, hyperparameters, and additional results for all three case studies are provided below.

\begin{table}[t]
\centering
\caption{Experimental design: each case study exercises a different modal logic.}
\label{tab:experiment_overview}
\resizebox{\linewidth}{!}{
\begin{tabular}{@{}clllll@{}}
\toprule
\textbf{\#} & \textbf{Experiment} & \textbf{Modal Type} & \textbf{Key Formula} & \textbf{SDE} & \textbf{Domain} \\ \midrule
1 & Swarm Hallucination & Epist.+Dox. & $B_{\text{r3}}(\Box\text{safe}) \land K_{\text{sw}}(\Diamond\text{collision})$ & Epist.+Dox. & Custom 2D Swarm \\
2 & Lorenz Chaotic & Temporal & $\Box\text{(bounded)} \land \Diamond\text{(visits\_lobe)}$ & Temporal & Lorenz-63 \\
3 & Safe Confinement & Deontic & $O(\Box\text{safe})$ & Temp.+Deon. & 2D Vessel \\
\bottomrule
\end{tabular}}
\end{table}

\subsection{Case Study 1: Multi-Agent Swarm Modal Logic}

The hallucination detection experiment requires three semantically distinct SDEs operating in a shared state space $\mathcal{W} \subseteq \mathbb{R}^{4 \times 5}$ (position and velocity for 5 rovers).  The temporal SDE models true physics, the epistemic SDE conditions on healthy sensor data, and the doxastic SDE uses Rover~3's faulty internal model.

\paragraph{Training dynamics.}  Fig.~\ref{fig:cs1_metrics} tracks the three key quantities across 5 seeds.  The doxastic bound $L_{B_3}(\Box\,\text{safe})$ rises above the 0.8 detection threshold within the first few epochs and stabilises, Rover~3 persistently believes it is safe.  Meanwhile, the epistemic bound $U_{K_{\text{sw}}}(\Diamond\,\text{collision})$ rises above 0.3, confirming the swarm collectively knows a collision path exists.  The Wasserstein belief--reality gap converges to $W = 0.39$, providing a continuous scalar measure of how far Rover~3's believed future diverges from the true future.

\begin{figure}[h]
\centering
\includegraphics[width=\linewidth]{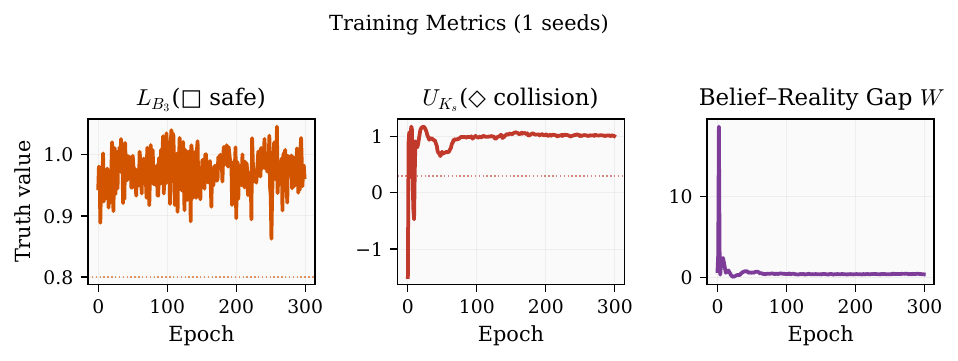}
\caption{Training dynamics (5 seeds, mean $\pm$ std).  Left: doxastic bound $L_{B_3}(\Box\,\text{safe})$ rises above the 0.8 threshold.  Centre: epistemic bound $U_{K_{\text{sw}}}(\Diamond\,\text{collision})$ rises above 0.3.  Right: Wasserstein belief--reality gap stabilises at $W = 0.39$.}
\label{fig:cs1_metrics}
\end{figure}

\paragraph{Where does hallucination arise?}  The answer is spatial: detection depends on which region of the arena Rover~3 occupies.  Fig.~\ref{fig:cs1_field_safety} shows the safety field as seen by the true model versus Rover~3's faulty model, the real chasm and the hallucinated chasm are in different locations.  Fig.~\ref{fig:cs1_modal_field} maps the resulting modal truth values over a $50 \times 50$ grid of starting positions.  Hallucination is strongest where the doxastic SDE predicts safety but the epistemic SDE predicts danger, exactly the trajectories passing through the true chasm.

\begin{figure}[h]
\centering
\includegraphics[width=\linewidth]{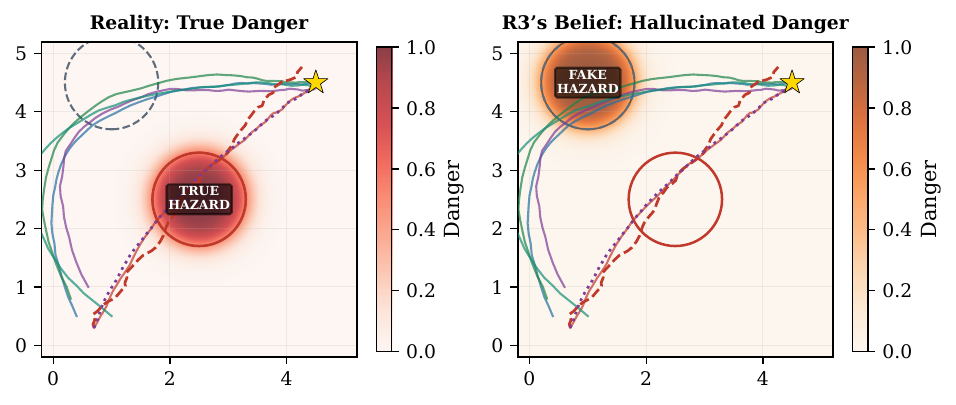}
\caption{Safety field over the arena.  Left: true danger gradient centred on the real chasm.  Right: danger field as believed by Rover~3, centred on the hallucinated (fake) chasm.}
\label{fig:cs1_field_safety}
\end{figure}

\begin{figure}[h]
\centering
\includegraphics[width=\linewidth]{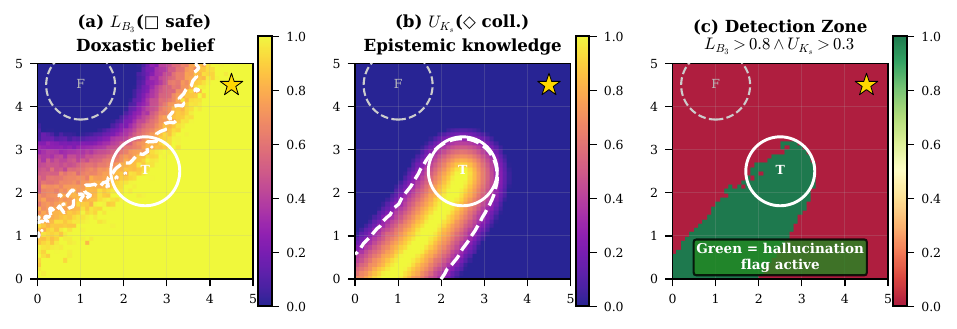}
\caption{Spatial modal truth field ($50 \times 50$ grid; raw robustness $L$, $U$).  Left: $L_{B_3}(\Box\,\text{safe})$.  Right: $U_{K_{\text{swarm}}}(\Diamond\,\text{collision})$.  White dashed contours show detection thresholds.}
\label{fig:cs1_modal_field}
\end{figure}

\paragraph{Sub-experiment B: Sustained escort.}  A rescuer rover learns to shield Rover~3 around the true chasm and guide it toward the goal.  Fig.~\ref{fig:cs1_escort_arena} shows the learned escort trajectory; Fig.~\ref{fig:cs1_escort_metrics} tracks the actual safety margin, believed safety, and the Wasserstein trust gap ($W = 1.53$).  The high trust gap indicates the rescuer must intervene precisely because Rover~3's own model does not perceive the danger.

\begin{figure}[h]
\centering
\includegraphics[width=\linewidth]{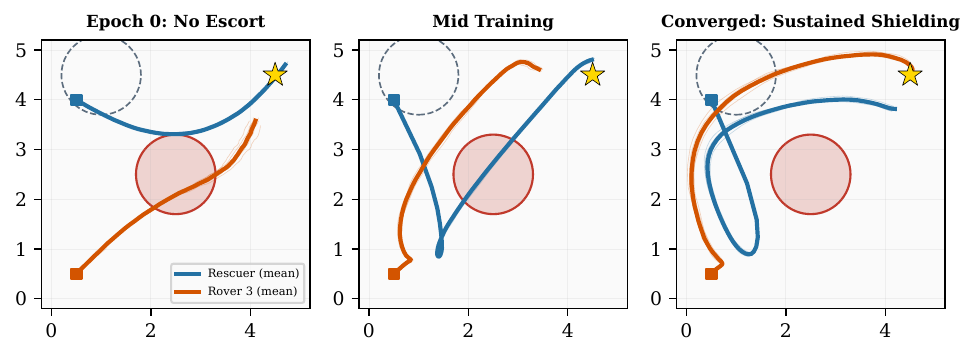}
\caption{Sub-Exp~B: 2-rover sustained escort.  The rescuer (blue) shields Rover~3 (orange) past the true chasm while R3 homes toward the goal (star).  Final actual safety: $0.82$.}
\label{fig:cs1_escort_arena}
\end{figure}

\begin{figure}[h]
\centering
\includegraphics[width=\linewidth]{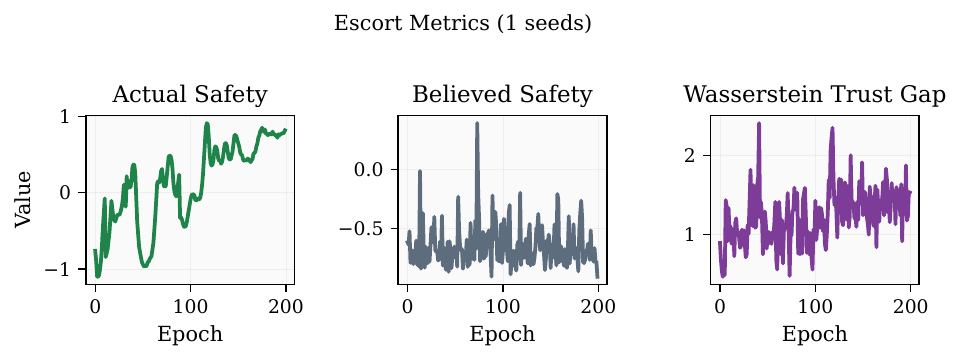}
\caption{Escort training metrics.  Left: actual safety margin improves.  Centre: believed safety remains high (R3 never perceives the true hazard).  Right: Wasserstein trust gap ($W = 1.53$).}
\label{fig:cs1_escort_metrics}
\end{figure}

\paragraph{Sub-experiment C: Swarm defense.}  Five agents jointly optimise five temporal logic objectives: scouts defend a VIP while evading enemies (Fig.~\ref{fig:cs1_defense_arena}).  This exercises the CMLNN's ability to handle multiple formulas over shared dynamics.  The final formula satisfaction values (Fig.~\ref{fig:cs1_defense_formulas}) show that safety and boundary constraints are met ($\Box$~in\_bounds: $0.16$, $\Box$~safe: $0.12$, $\Box$~follow: $-0.03$), while the offensive and exploration objectives exhibit the expected difficulty ($\Diamond$~attack: $-0.20$, $\Diamond$~explore: $-0.47$).

\begin{figure}[h]
\centering
\includegraphics[width=\linewidth]{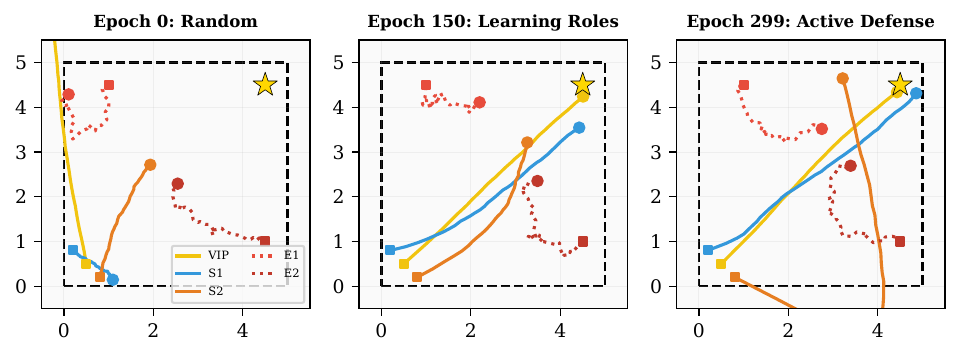}
\caption{Sub-Exp~C: 5-agent swarm defense (VIP protection).  Two scouts defend a VIP against two enemies while the VIP explores toward the goal.}
\label{fig:cs1_defense_arena}
\end{figure}

\begin{figure}[h]
\centering
\includegraphics[width=\linewidth]{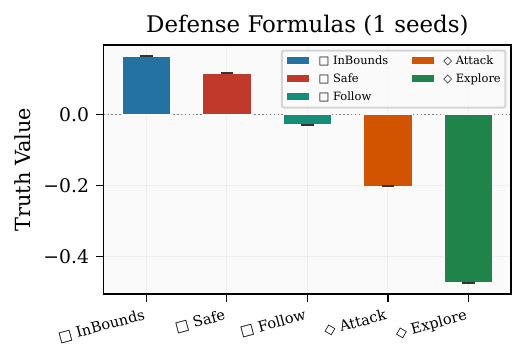}
\caption{Defense formula satisfaction.  Five temporal logic objectives jointly optimised.}
\label{fig:cs1_defense_formulas}
\end{figure}

\subsection{Case Study 2: Lorenz Chaotic System (LINNs)}

The Lorenz-63 system provides a challenging benchmark for modal logic: the formula $\Box(\text{bounded}) \land \Diamond(\text{visits\_lobe})$ requires all trajectories to remain bounded (necessity) while some must visit the secondary lobe (possibility).  This is impossible for deterministic models (quantifier collapse) and requires the SDE's stochastic branching.

\paragraph{Training protocol.}  Training proceeds in two phases (Fig.~\ref{fig:cs2_training}).  Phase~1 pre-trains the SDE drift on trajectory MSE.  Phase~2 ramps in the LINN logic loss via a scheduled $\lambda$, gradually shifting the objective from trajectory matching toward logical consistency.  All five modal formulas converge.

\begin{figure}[h]
\centering
\includegraphics[width=\linewidth]{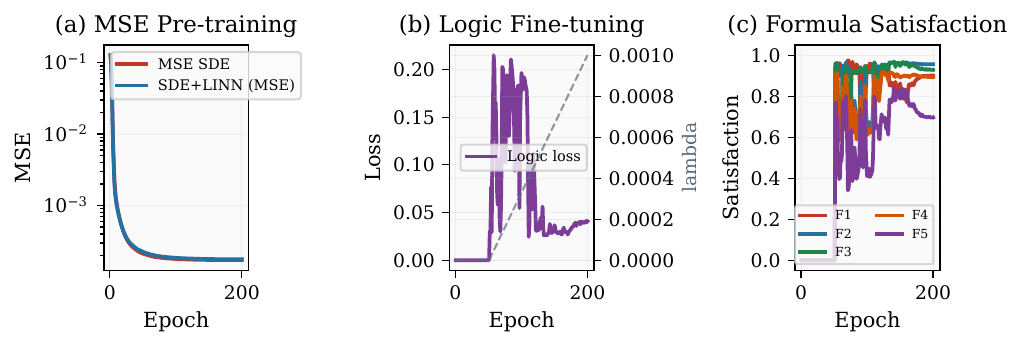}
\caption{Training dynamics.  (a)~MSE pre-training.  (b)~Logic loss and $\lambda$ ramp.  (c)~Formula satisfaction converges.}
\label{fig:cs2_training}
\end{figure}

\paragraph{Quantifier divergence and structural properties.}  The defining feature of a non-collapsed modal logic is that $\Diamond\phi - \Box\phi > 0$.  Fig.~\ref{fig:cs2_structural}(a) shows this gap over the forecast horizon: SDE+LINN tracks the ground truth while deterministic models remain at zero.  Fig.~\ref{fig:cs2_structural}(b) confirms that SDE+LINN recovers the balanced lobe visitation (${\sim}0.5/0.5$) characteristic of the true Lorenz attractor.

\begin{figure}[h]
\centering
\begin{minipage}[b]{0.48\linewidth}\centering
  \includegraphics[width=\linewidth]{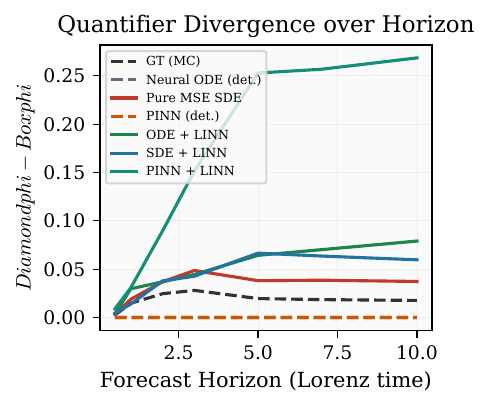}\\(a) Quantifier divergence
\end{minipage}\hfill
\begin{minipage}[b]{0.48\linewidth}\centering
  \includegraphics[width=\linewidth]{i/cs2_lobe_visitation.pdf}\\(b) Lobe visitation
\end{minipage}
\caption{(a)~$\Diamond\phi - \Box\phi$ over forecast horizon: SDE+LINN tracks the ground truth while deterministic models remain at zero (quantifier collapse).  (b)~Lobe visitation: SDE+LINN matches the ground truth ${\sim}0.5/0.5$ balance.}
\label{fig:cs2_structural}
\end{figure}

\paragraph{Soundness verification.}  A sound modal operator must satisfy $L_\Box \leq \text{GT}_\Box$ and $U_\Diamond \geq \text{GT}_\Diamond$ in the same bounded robustness scale (Sec.~\ref{sec:method:modal_ops}).  Fig.~\ref{fig:cs2_soundness} confirms SDE+LINN satisfies both conditions, while pure MSE models are overconfident.

\begin{figure}[h]
\centering
\includegraphics[width=\linewidth]{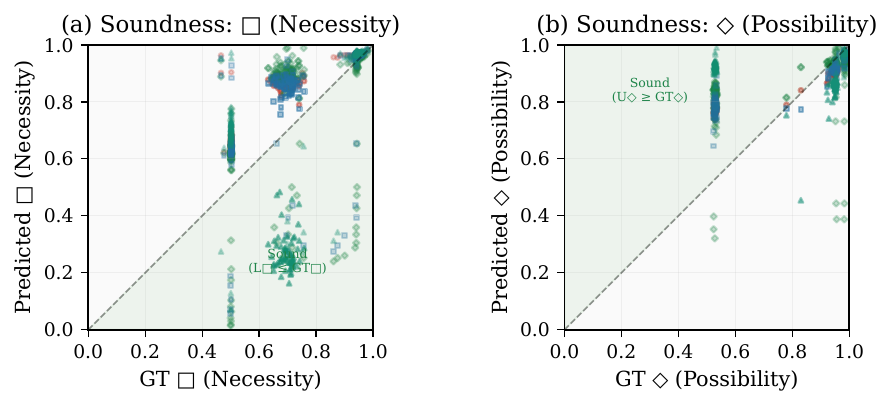}
\caption{Soundness verification: predicted $\Box \leq$ GT $\Box$ (left) and predicted $\Diamond \geq$ GT $\Diamond$ (right).  SDE+LINN satisfies both.}
\label{fig:cs2_soundness}
\end{figure}

\paragraph{Full model comparison.}  Fig.~\ref{fig:cs2_phase} shows $x$--$z$ phase portraits for all eight models.  Deterministic models (Neural ODE, PINN) produce single trajectories; PINNs diverge from the attractor basin despite knowing the true equations.  SDE+LINN produces the richest two-lobe attractor coverage, closest to the ground truth, demonstrating that logical constraints can recover global structural properties that local dynamics alone cannot enforce.

\begin{figure}[h]
\centering
\includegraphics[width=\linewidth]{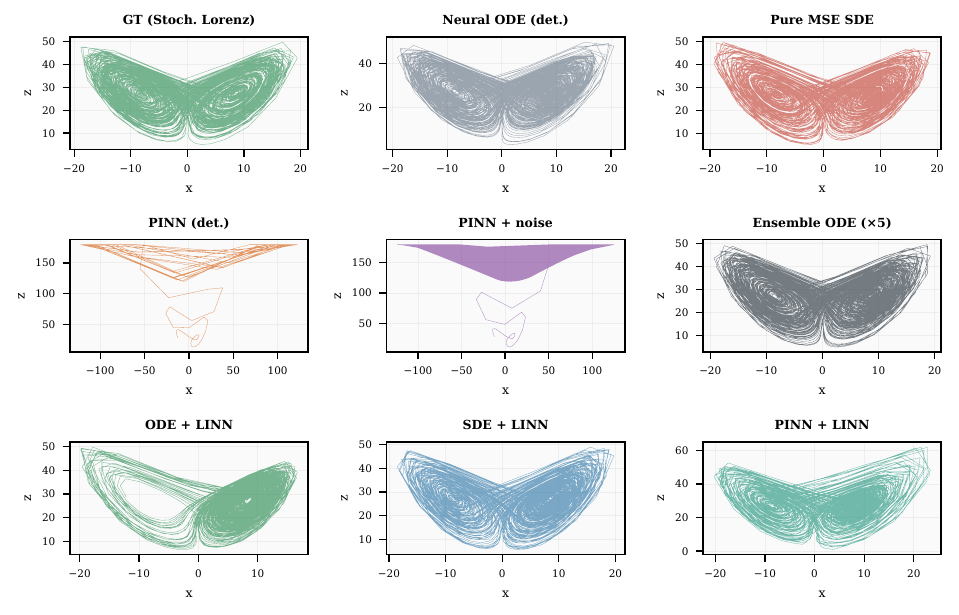}
\caption{Phase portraits ($x$--$z$ projection) for all eight models plus ground truth.  SDE+LINN produces the richest two-lobe attractor coverage.}
\label{fig:cs2_phase}
\end{figure}

\subsection{Case Study 3: Safe Confinement Dynamics}

The deontic case study demonstrates that a logical specification, $O(\Box_{[0,T]}\text{safe})$, can serve as a complete training objective for learning safe dynamics, without reward engineering or hand-crafted control laws.

\paragraph{Overview dashboard.}  Fig.~\ref{fig:cs6_overview} provides a four-panel view of the learned system: trajectory bundles in 2D cross-section (temporal paths exit; deontic paths stay confined), modal truth values over time, 3D toroidal embedding, and a quiver plot of the learned drift field.  The deontic drift curves inward at the vessel wall, creating a learned restoring force.

\begin{figure}[h]
\centering
\includegraphics[width=\linewidth]{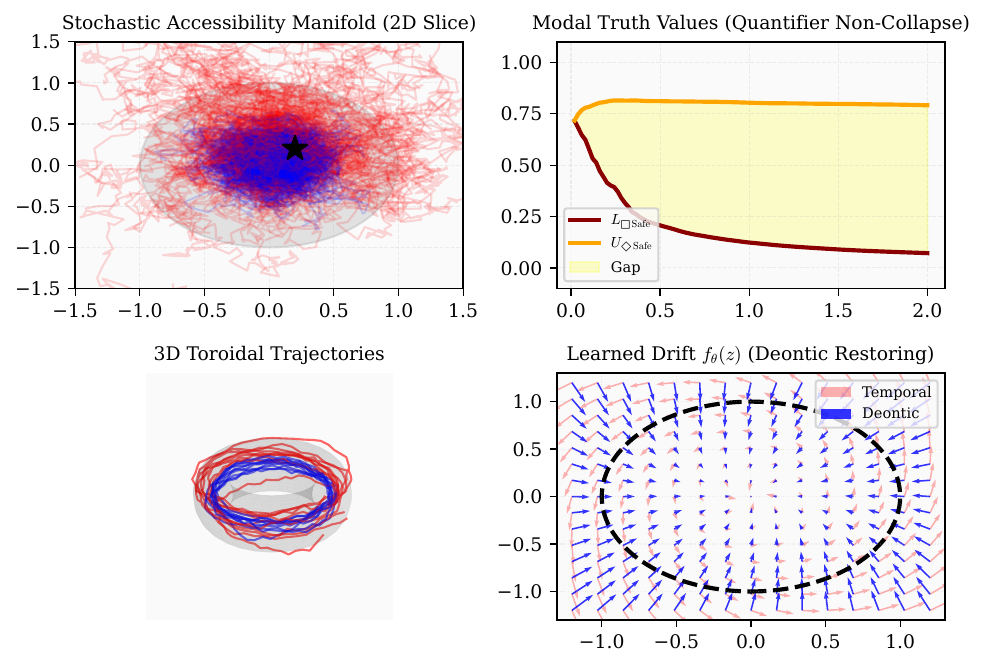}
\caption{Overview dashboard.  Top-left: temporal (red) vs deontic (blue) trajectory bundles.  Top-right: modal truth values and quantifier gap.  Bottom-left: 3D toroidal embedding.  Bottom-right: learned drift fields; the deontic drift curves inward at the wall.}
\label{fig:cs6_overview}
\end{figure}

\paragraph{Per-path safety and exit analysis.}  Fig.~\ref{fig:cs6_safety} tracks the safety value $\phi_{\text{safe}}(z_t)$ along individual sample paths.  Under temporal dynamics, safety degrades as paths drift outward; under the learned deontic dynamics, safety is maintained throughout the horizon.  Fig.~\ref{fig:cs6_distributions} quantifies this: the temporal exit fraction rises steadily while the deontic exit fraction remains at zero, and the radial distributions confirm that deontic paths are concentrated well inside the vessel boundary.

\begin{figure}[h]
\centering
\includegraphics[width=\linewidth]{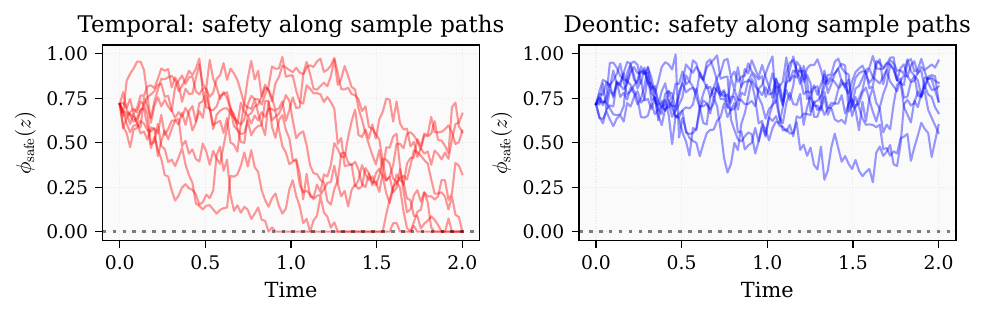}
\caption{Per-path safety $\phi_{\text{safe}}(z_t)$.  Left: temporal dynamics, where safety degrades over time.  Right: deontic dynamics, where safety is maintained throughout the horizon.}
\label{fig:cs6_safety}
\end{figure}

\begin{figure}[h]
\centering
\begin{minipage}[b]{0.48\linewidth}\centering
  \includegraphics[width=\linewidth]{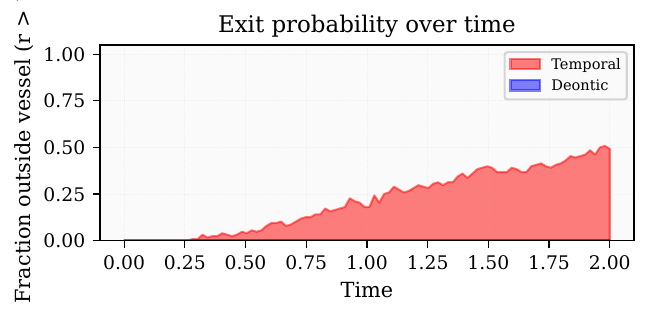}\\(a) Exit probability over time
\end{minipage}\hfill
\begin{minipage}[b]{0.48\linewidth}\centering
  \includegraphics[width=\linewidth]{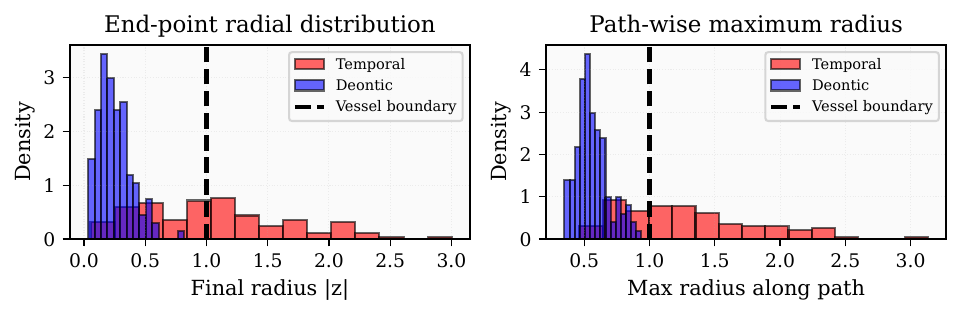}\\(b) Radial distributions
\end{minipage}
\caption{(a)~Exit fraction: temporal rises steadily while deontic remains at zero.  (b)~Radial distributions: deontic paths are concentrated inside the vessel.}
\label{fig:cs6_distributions}
\end{figure}

\paragraph{Training dynamics.}  Fig.~\ref{fig:cs6_training} shows the multi-seed training curves: $L_{\Box}$ and $U_{\Diamond}$ increase as the deontic SDE learns confinement, the exit fraction drops to zero, and the total loss converges.

\begin{figure}[h]
\centering
\includegraphics[width=\linewidth]{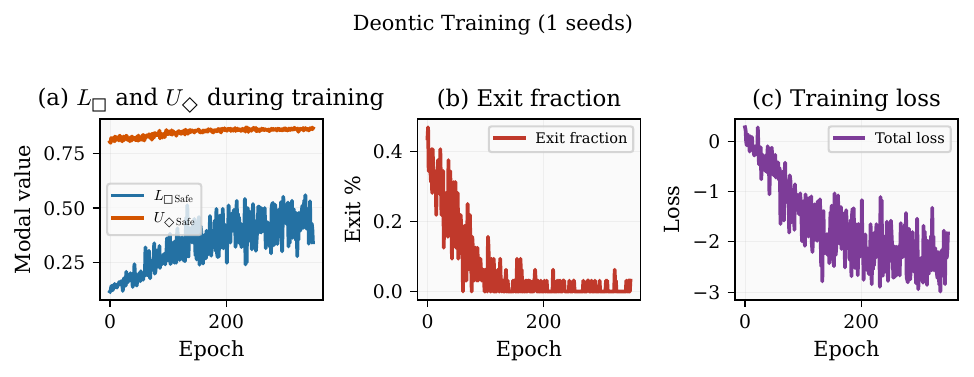}
\caption{Deontic training curves (multi-seed).  (a)~$L_{\Box}$ and $U_{\Diamond}$ increase.  (b)~Exit fraction drops to zero.  (c)~Loss converges.}
\label{fig:cs6_training}
\end{figure}

\paragraph{Learned drift magnitude.}  Fig.~\ref{fig:cs6_heatmap} shows the magnitude of the learned restoring force.  The bright boundary ring in the deontic drift field (right panel) confirms that the network has learned a strong restoring force precisely at the vessel boundary, and this structural property generalises even to mid-radius regions that were not explicitly trained on.

\begin{figure}[h]
\centering
\includegraphics[width=\linewidth]{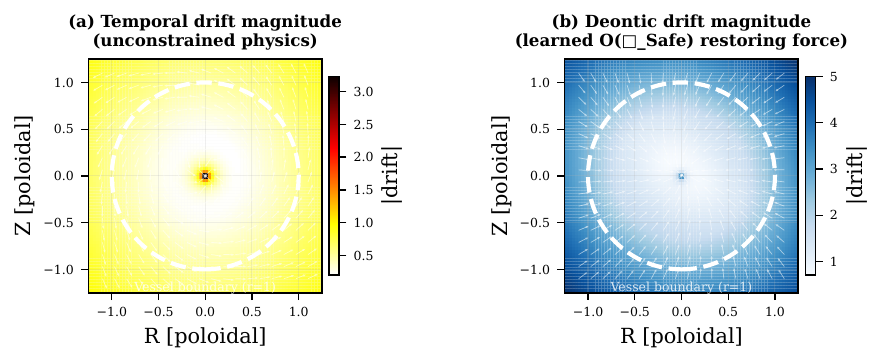}
\caption{Drift magnitude heatmaps.  Left: temporal drift (unconstrained physics).  Right: deontic drift, showing a learned strong restoring force (bright ring) near the vessel boundary.}
\label{fig:cs6_heatmap}
\end{figure}
\end{document}